\documentclass[11pt]{article}

\usepackage[utf8]{inputenc}
\usepackage[T1]{fontenc}
\usepackage[english]{babel}

\usepackage{fullpage}
\usepackage{microtype}
\usepackage{pdfsync}
\usepackage[margin=1in]{geometry}

\usepackage{hyperref}
\usepackage{url}

\usepackage{enumitem}

\usepackage{amsmath,amssymb,amsfonts,amsthm}
\usepackage{mathabx}

\usepackage{mathtools}
\usepackage{algorithm}
\usepackage[noend]{algpseudocode}
\usepackage{tikz}
\usetikzlibrary{shapes,decorations.pathreplacing}

\newtheorem{lemma}{Lemma}
\newtheorem{theorem}{Theorem}
\newtheorem{corollary}{Corollary}

\theoremstyle{definition}
\newtheorem{definition}{Definition}
\newtheorem{observation}{Observation}
\theoremstyle{remark}

\usepackage{authblk}

\newcommand{\E}{\textnormal{\textrm{E}}}
\newcommand{\HH}{\ensuremath{\mathcal{H}}}
\DeclarePairedDelimiterX{\infdivx}[2]{}{}{ \left(#1\;\delimsize\middle\|\;#2\right) }
\DeclareMathOperator*{\D}{D}
\newcommand{\infdiv}{\D\infdivx}
\DeclareMathOperator*{\HE}{H}
\DeclareMathOperator*{\argmin}{arg\,min}
\newcommand{\TT}{\ensuremath{\mathcal{T}}}
\newcommand{\dist}{\ensuremath{\textnormal{dist}}}
\newcommand\numberthis{\addtocounter{equation}{1}\tag{\theequation}}

\usepackage[textsize=tiny]{todonotes}
\newcommand{\matodo}[1]{}
\newcommand{\tatodo}[1]{}
\newcommand{\rptodo}[1]{}

\title{
Parameter-free Locality Sensitive Hashing for Spherical Range Reporting%
\footnote{The research leading to these results has received funding from the European Research Council under the European Union's 7th Framework Programme (FP7/2007-2013) / ERC grant agreement no. 614331.}
}

\author{Thomas D. Ahle, Martin Aumüller, and Rasmus Pagh}
\affil{IT University of Copenhagen, Denmark, \{thdy, maau, pagh\}@itu.dk }

\begin{document}

\maketitle

\setcounter{page}{0}
\begin{abstract}
We present a data structure for {\em spherical range reporting} on a point set $S$, i.e., reporting all points in $S$ that lie within radius $r$ of a given query point $q$ (with a small probability of error).
Our solution builds upon the Locality-Sensitive Hashing (LSH) framework of Indyk and Motwani, which represents the asymptotically best solutions to near neighbor problems in high dimensions.
While traditional LSH data structures have several parameters whose optimal
values depend on the distance distribution from $q$ to the points of $S$ (and in
particular on the number of points to report), our data structure is essentially
parameter-free and only takes as parameter the space the user is
willing to allocate.
Nevertheless, its expected query time basically matches that of an LSH data structure
whose parameters have been {\em optimally chosen for the data and query} in
question under the given space constraints.
In particular, our data structure provides a smooth trade-off between hard queries (typically addressed by standard LSH parameter settings) and easy queries such as those where the number of points to report is a constant fraction of~$S$, or where almost all points in~$S$ are far away from the query point.
In contrast, known data structures fix LSH parameters based on certain parameters of the input alone.

The algorithm has expected query time bounded by $O(t (n/t)^\rho)$, where $t$ is the number of points to report and $\rho\in (0,1)$ depends on the data distribution and the strength of the LSH family used.
We further present a parameter-free way of using multi-probing, for LSH families that support it, and show that for many such families this approach allows us to get expected query time close to $O(n^\rho+t)$, which is the best we can hope to achieve using LSH.
The previously best running time in high dimensions was $\Omega(t n^\rho)$, achieved by traditional LSH-based data structures where parameters are tuned for outputting a single point within distance $r$.
Further, for many data distributions where the intrinsic dimensionality of the point set close to~$q$ is low, we can give improved upper bounds on the expected query time.

\end{abstract}

\thispagestyle{empty}

\newpage
\setcounter{page}{1}
\section{Introduction}\label{sec:introduction}


\emph{Range search} is a central problem in computational geometry
\cite{AgarwalE99}. Given a set $S$ of $n$ points in~$\mathbb{R}^d$, build a data
structure that answers queries of the following type: Given a region $R$ (from a
predefined class of regions), \emph{count} or \emph{report} all points from 
$S$ that belong to $R$. Examples for such classes of regions are simplices \cite{Matousek94}, halfspaces
\cite{ChazelleLM08}, and spheres \cite{arya2010unified}. 

In this paper we place our main focus on the {\em spherical range reporting problem} (SRR): Given a distance parameter $r$
and a point set $S$, build a data structure that supports the following queries: Given a point $q$, report all points in $S$ within distance $r$ from $q$. This problem is closely related to 
\emph{spherical range counting} (``return the number of points'') and \emph{spherical range emptiness} (``decide whether there is a point at distance at most  $r$'').
Solving spherical range searching problems exactly, i.e., for $\varepsilon = 0$, and in time that is truly sublinear in the point set size $n=|S|$ seems to require space exponential in the dimensionality of the point set $S$.
This phenomenon is an instance of the \emph{curse of dimensionality}, and is supported by popular algorithmic hardness conjectures (see \cite{AlmanW15,Williams05}).

For this reason, most algorithms for range searching problems involve approximation of distances: For some approximation parameter $c>1$ we allow the data structure to only distinguish between distance~$\leq r$ and~$> cr$, while points at distance in between can either be reported or not. We refer to this relaxation as $c$-approximate SRR. Approximate range reporting and counting problems were for example considered by Arya et al. in \cite{arya2010unified}, by Indyk in his Ph.D. thesis
\cite{Indyk00} as 
``enumerating/counting point locations in equal balls'' and by Andoni in his Ph.D. thesis \cite{Andoni09} as ``randomized R-near neighbor reporting''.  
In low dimensions, tree-based approaches allow us to build efficient data structures with space usage $\tilde{O}(n\gamma^{d-1}(1 + (c-1)\gamma^2))$ and query time
$\tilde{O}(1/((c-1)\gamma)^{d-1})$ for a trade-off parameter $\gamma \in [1, 1/(c-1)]$ for an approximation factor
$1 < c \leq 2$, see \cite{arya2010unified}. 
The exponential dependency of time and/or space on the dimension makes these algorithms inefficient in high dimensions. 

Our approach uses the \emph{locality-sensitive hashing}
(LSH) framework \cite{IndykM98} which we will introduce in 
Section~\ref{sec:prelim}.  Using this technique to solve 
SRR is not new: Both Indyk \cite{Indyk00} and Andoni \cite{Andoni09} described 
extensions of the general LSH framework to solve this problem. As we will show, 
their approaches yield running times of $\Omega(tn^\rho)$, where $t$ is 
the number of elements at distance at most $cr$ from the query and $\rho \in (0, 1)$ is a parameter that depends
on the distance $r$, the approximation factor $c$, and the LSH family used to build
the data structure. When the output size $t$ is large this leads to running times of $\Omega(n^{1 + \rho})$,  which is worse
than a linear scan! 
Indyk~\cite{Indyk00} also describes a reduction from spherical range counting to 
the $(c,r)$-approximate near neighbor problem that asks to report a \emph{single} point from the result set of $c$-approximate SRR. 
The reduction uses  $O(\log^2 n/(c-1)^3)$ queries of independently built $(c,r)$-ANN data structures, giving 
a running time of $O(n^\rho\log^2 n/(c-1)^3)$. 
Building upon Indyk's technique, Chazelle et al.
\cite{ChazelleLM08} proposed a data structure that solves approximate halfspace
range queries on the unit sphere by applying a dimension reduction technique to
Hamming space. All of these algorithms use a standard LSH index data structure in a black-box fashion. 
We propose a data structure that is almost similar to a standard LSH data structure, but query it 
in an adaptive way. Our guarantees are probabilistic in the sense that each close point is with constant probability present
in the output. 

Using LSH-based indexes for range reporting means that we get all points closer than distance~$r$ with a certain
probability, as well as some fraction of the points with distance in the range $(r, cr)$.
When $c$ is large, this can have negative consequences for performance: a query could report nearly every point in the data set, and any performance gained from approximation is lost.
When the approximation factor $c$ is set close to $1$, data structures working in 
high dimensions usually need many independent repetitions to find points at distance~$r$.
This is another issue with such indexes that makes range reporting hard: very close points show up in every repetition, and we need to remove these duplicates.

\begin{figure}[t]
\centering
\includegraphics[width=.7\textwidth]{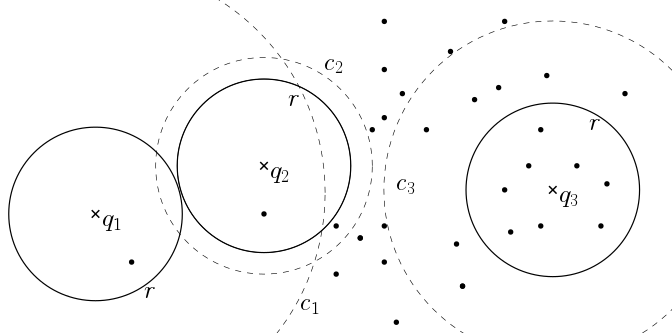}
\caption{Three queries at radius $r$ with points $q_1, q_2,$ and $q_3$ (crosses) in a data set (dots). 
Dashed circles around the queries show how far the radius-$r$ sphere can be stretched such that the number of points between radius $r$ and radius $c_i r$ equals the number of points in radius $r$. We see that queries $q_1$ and $q_3$ allow for a large stretch, while query $q_2$ has a small stretch. 
}
\label{fig:example}
\end{figure}

\begin{figure}[t]
\vspace{1em} 
\centering
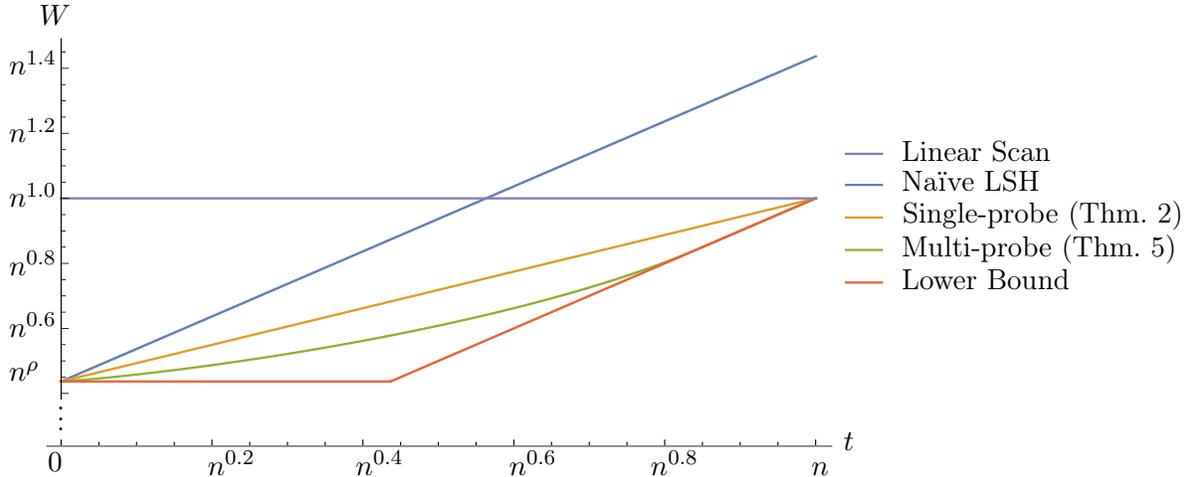
\caption{Overview of the running time guarantees
    of the proposed algorithms ``Adaptive Single-probe'' and 
    ``Adaptive Multi-probe'' for collision probabilities $p_1 = .8$ and 
    $p_2 = .6$ in $d$-dimensional Hamming space such that $\rho \approx .436$. 
    The $x$-axis shows the value of $t$ as a function of $n$, the $y$-axis
    shows the expected work $W$, i.e., the expected number of points the 
    algorithm retrieves from the hash tables.  
    For comparison, we plotted the target time of $O(n^\rho + t)$, the running
    time $O(t n^\rho)$ of a naïve LSH approach, and the running time $O(n)$ of
a linear scan.}
\label{fig:work/logn}
\end{figure}

The natural approach to overcome the difficulties mentioned above is to choose the approximation factor $c$ such that the cost of duplicated points roughly equals the cost of dealing with far points.
For LSH-based algorithms, many papers explain an offline approach of finding the ``optimal'' value of $c$ for a data set \cite{andoni20052,Bawa05,Dong08} which usually envolves sampling or assumptions on the data distribution. 
However, the best value of $c$ depends not only on the data set, but also on the \emph{query}. This situation is depicted in Figure~\ref{fig:example}. In this paper, we provide a query algorithm that adapts to the input and finds a near-optimal $c$ at {\em query time}.  We manage to do this in time proportional to the number of points eventually returned for the optimal parameters, making the search essentially free. 

\paragraph{Output-sensitivity.}
To illustrate the improvement over standard fixed parameter LSH, we propose hard data sets 
for spherical range reporting. 
In these data sets, we pick $t-1$ very close points that show up in almost every repetition, one point at distance $r$, and the remaining points close to distance $cr$.
In this case LSH would need $\Theta(n^\rho)$ repetitions to retrieve the point at distance $r$ with constant probability, where e.g.~$\rho=1/c$ in Hamming space~\cite{IndykM98} and $\rho=1/c^2$ in Euclidean space~\cite{AndoniI06}.
This means that the algorithm considers $\Theta(tn^\rho)$ candidate points, which could be as large as $\Theta(n^{1+\rho})$ for large $t$. In Section~\ref{sec:algorithm} we describe and analyze two algorithms \emph{Adaptive Single-probe} and 
\emph{Adaptive Multi-probe} that mitigate this problem. 
The basic idea is that these algorithms ``notice'' the presence of many close points, and respond by setting $c$ more lenient, allowing for $t$ far points being reported per repetition in addition to the $t$ close points.
This in turn allows doing only $\Theta((n/t)^\rho)$ repetitions, giving a total candidate set of size $\Theta(t(n/t)^\rho)$, which is never larger than $n$. In general, the number of points between distance $r$ and $cr$ have a linear influence in these running times. 
This is made precise in Section~\ref{sec:results}. 

\paragraph{Multi-probing.}
When we stick to the LSH framework, the ideal solution would never consider a candidate set larger than $\Theta(n^\rho + t)$, giving the optimal output sensitive running time achievable by (data independent) LSH.
In order to get closer to this bound, we analyze the \emph{multi-probing approach} for LSH data structures, introduced in \cite{Panigrahy06} and further developed in~\cite{Lv07}.
The idea is that LSH partitions the space in many buckets, but usually only examines the exact bucket in which the query point falls in each repetition. Multi-probing considers all buckets ``sufficiently correlated'' with the query bucket to 
increase the likelihood of finding close points.
To our knowledge,  multi-probing has always been applied in order to save memory by allowing a smaller number of repetitions to be made and trading this for an increase in query time.
Our motivation is different: We want to take advantage of the fact that each of the very close points can only be in one bucket per repetition.
Hence by probing multiple buckets in each repetition, \emph{we not only save memory, but also gain a large improvement in the dependency on~$t$ in our running time}.
We do this by generalizing the adaptive single-probe algorithm to not only find the optimal~$c$ for a query, but also the optimal number of buckets to probe.
As we show in Section~\ref{sec:algorithm}, we are able to do this in time negligible compared to the size of the final candidate set, making it practically free.
The algorithm works for any probing sequence supplied by the user, but in Section~\ref{sec:hamming:multiprobing} we provide a novel probing sequence for Hamming space and show that it always improves the query time compared to the non-multi-probing variant. For certain regimes of $t$, we show that the running time matches the target time $O(n^\rho + t)$. An overview 
of the exact running time statements of the algorithms proposed here with a comparison to standard LSH, a linear scan, and 
the optimal running time for LSH-based algorithms is depicted in Figure~\ref{fig:work/logn}. 
\paragraph{Techniques.}
The proposed data structure is very similar to a standard LSH data structure. While such a data structure usually uses only one particular concatenation length $k$ of hash functions, we build the data structure for all lengths $1, \ldots, k$. At query time, we do an efficient search over the parameter space to find the provably best level. The algorithm then retrieves only the candidates from this level and filters far away points and duplicates.
The reason we are able to do an efficient search over the parameter space is that certain parts of the output size can be estimated very quickly when storing the size of the hash table buckets in the LSH data structure.
For example, when considering very large $c$, though the output may be large, there are only few repetitions to check.
Gradually decreasing $c$, we will eventually have to check so many repetitions that the mere task of iterating through them would be more work than scanning through the smallest candidate set found so far.
Since the number of repetitions for each value of $c$ grows geometrically, it ends up being bounded by the last check, which has size not larger than the returned candidate set.
For multi-probing it turns out that a similar strategy works, even though the search problem is now 
two-dimensional. 

\paragraph{Additional Related Work.} 
Our approach to query-sensitivity generalizes and extends the recent work of Har-Peled and Mahabadi~\cite{Har-PeledM15} which considers approximate near neighbors.  Our method applies to every space and metric supported by the LSH framework while~\cite{Har-PeledM15} is presented for Hamming space.

The proposed single-probing algorithm can be thought of as an adaptive query algorithm on the trie-based LSH forest introduced by Bawa et al. in~\cite{Bawa05} for the related approximate $k$-nearest neighbor problem. (The authors of \cite{Bawa05} make significant assumptions
on the distance distribution of approximate nearest neighbors). The algorithm proposed there always looks at all $n^{f(c)}$ repetitions where $f(c)$ depends on the largest distance $r$ supported by the algorithm and the approximation factor. It collects points traversing tries synchronously in a bottom-up fashion. By looking closer at the guarantees of LSH functions, we show that one can gradually increase the number of repetitions to look at and find the best level to query
directly. We hope that the insights provided here will shed new light on solving approximate nearest neighbors beyond using standard reductions as in \cite{Har-PeledIM12}.

Combining results of very recent papers \cite{laarhoven2015tradeoffs,Christiani16,andoni2016lower} on space/time-tradeoffs make it possible to achieve running times that are similar to our results with respect to multi-probing. We give a detailed 
exposition in Appendix~\ref{sec:appendix2} and plan to include these results in the 
final version.

\rptodo{I think this section can go somewhere else, perhaps to an appendix with a forward pointer from the introduction.}
\matodo{I commented it out for now.}

\section{Preliminaries}\label{sec:prelim}

Let $(X, \dist)$ be a metric space over $X$ with distance function $\dist$.
In this paper, the space usually does not matter; only the multi-probing sequence in Section~\ref{sec:hamming:multiprobing}
is tied to Hamming space. 

\begin{definition}[Spherical Range Reporting, SRR]
   Given a set of points $S \subseteq X$ and a number $r \geq 0$, construct a
   data structure that supports the following queries: Given a point $q \in X$,
   report each point $p \in S$ with $\text{dist}(p, q) \leq r$ with constant probability.
\end{definition}
Note that we consider the exact version of SRR but allow point-wise probabilistic guarantees. 


\begin{definition}[Locality-Sensitive Hash Family, \cite{Charikar02}]
   A \emph{locality-sensitive hash family} $\HH$ is family of functions $h\colon X \to R$, such that for each pair $x, y \in X$ and a random $h\in\HH$, whenever $\text{dist}(q,x)\le\text{dist}(q,y)$ we have $\Pr[h(q)=h(x)]\ge\Pr[h(q)=h(y)]$, 
   for arbitrary $q \in X$.
\end{definition}

Usually the set $R$ is small, like the set $\{0,1\}$.
Often we will concatenate multiple independent hash functions from a family, getting functions $h_k\colon X \to R^k$.  We call this a hash function at level $k$.

Having access to an LSH family $\HH$ allows us to build a data structure with the following properties. 
\begin{theorem}[{\cite[Theorem 3.4]{Har-PeledIM12}}]
   Suppose there exists an LSH family such that
   $\Pr[h(q)=h(x)] \ge p_1$ when $\text{dist}(q,x)\le r$ and
   $\Pr[h(q)=h(x)] \le p_2$ when $\text{dist}(q,x)\ge cr$ with $p_1 > p_2$,
   for some metric space $(X, \dist)$ and some factor $c > 1$.
   Then there exists a data structure such that for a given query~$q$, it returns with constant probability a point within distance $cr$, if there exists a point within distance~$r$.
   The algorithm uses $O(d n + n^{1 + \rho})$ space and $O(n^\rho)$ hash function evaluations per query, where $\rho = \frac{\log 1/p_1}{\log 1/p_2}$.
   \label{thm:lsh}
\end{theorem}
It is essential for understanding our algorithms to know how the above data structure works. For the 
convenience of the reader we  
provide a description of it in Appendix~\ref{sec:proof:lsh}.

In this paper, we state the $\rho$-parameter as used in Theorem~\ref{thm:lsh} as 
a function $\rho(r, c)$ such that $\rho(r, c) = \frac{\log (1/p(r))}{\log (1/p(cr))}$,
where $p(\Delta)$ is the probability that two points at distance $\Delta$ collide. (The 
probability is over the random choice of the LSH function.)
We omit the parameters when their value is clear from the context.

A common technique when working with LSH is multi-probing \cite{Lv07,Panigrahy06,Dong08,Andoni15,Kapralov15}. 
The idea is that often the exact bucket $h_k(q)$ does not have a much higher collision probability with close points than some ``nearby'' bucket $\sigma(h_k(q))$.
Hence we probe multiple buckets in each repetition to reduce the space needed for storing repetitions.
In this paper we are going to show how this approach can give not just space improvements, but also large improvements in query time for SRR.


The LSH framework can easily be extended to solve SRR.
We just report all the points that are in
distance at most $r$ from the query point in the whole candidate set retrieved
from all tables $T_1, \ldots, T_L$ \cite{Andoni09}. For the remainder of this paper, we 
will denote the number of points retrieved in this way by $W$ (``work''). 
It is easy to see that this change to the query algorithm
would already solve SRR with the
guarantees stated in the problem definition. 
However, we will see in Section~\ref{sec:results} that its running time might be as large as $O(n^{1 + \rho})$, worse than a linear scan over the data set.

\section{Data Structure}\label{sec:datastructure}
We extend a standard LSH data structure in the following way.
\matodo{TODO:
Thomas, please make sure this makes sense.
}
\begin{definition}[Multi-level LSH]
Assume we are given a set $S \subseteq X$ of $n$ points, two parameters~$r$ and~$L$, and access to an 
LSH family $\HH$ that maps from $X$ to $R$. Let 
$\text{reps}(k)=\lceil p_1^{-k}\rceil$ where
$p_1$ is the probability that points at distance $r$ collide under random choice of $h \in \HH$. 
Let $K$ be the largest integer $k$ such that $\text{reps}(k) \leq L$. A
\emph{Multi-level LSH data structure} for $S$ is set up in the following way: 
For each $k \in \{0, \ldots, K\}$ choose functions $g_{k, i}$ for $1 \leq i \leq \text{reps}(k)$ 
from $\HH$ independently at random.  
Then, for each $k \in \{0, \ldots, K\}$, build $\text{reps}(k)$ hash tables
$T_{k, i}$ with $1 \leq i \leq \text{reps}(k)$. 
For a fixed pair $k \in \{0, \ldots, K\}$ and $i \in \{1, \ldots,
\text{reps}(k)\}$, and each $x \in X$,  concatenate hash values $(g_{1, i}(x),
\ldots, g_{k, i}(x)) \in R^k$
to obtain the hash code $h_{k, i}(x)$. Store references to all points in $S$
in table $T_{k, i}$ by applying $h_{k, i}(x)$. 
For a point $x \in X$, and for integers $0 \leq k \leq K$ and $1 \leq i \leq
\text{reps}(k)$, 
we let $|T_{k, i}(x)|$ be the number of points in bucket
$h_{k, i}(x)$ in table~$T_{k, i}$. We assume this value can be retrieved in constant time.
\label{def:multi:level:lsh}
\end{definition}
In contrast to a standard LSH data structure, we only accept the number of
repetitions as an additional parameter.  
The value $K$ is chosen such that the number of repetitions available suffices to obtain a close point at distance $r$ with constant probability, cf.~Appendix~\ref{sec:proof:lsh}.
This is ensured by the repetition count for all levels $0, \ldots, K$.
The space usage of our data structure is $O(n\sum_{0 \leq k \leq K}p_1^{-k}) =
O(n p_1^{-K)}) = O(nL)$.
Hence multiple levels only add a constant overhead
to the space consumption compared to a standard LSH data structure for level
$K$. Figure~\ref{fig:datastructure} provides a visualization of the data
structure.

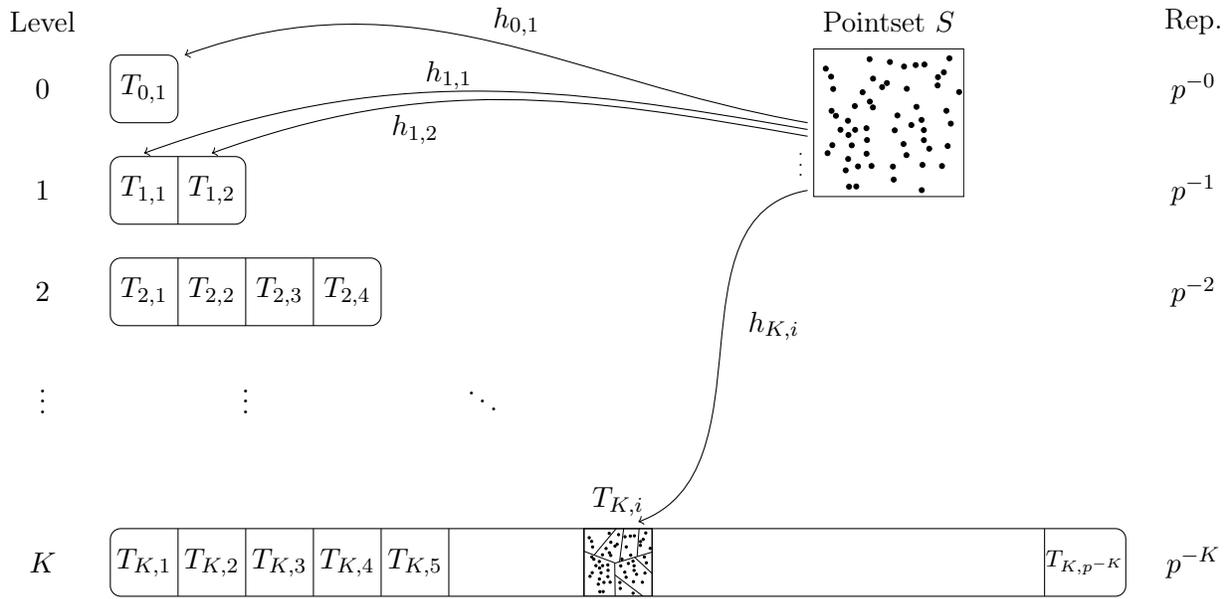
\begin{figure}[t]
\begin{tikzpicture}[scale=0.9]
\draw[draw, rounded corners] (0,0) rectangle (1,1);
\draw[draw, rounded corners] (0,-1.5) rectangle (2,-0.5);
\draw[draw, rounded corners] (0,-3) rectangle (4,-2);
\node[] at (7.5, -6.5) {\includegraphics[height=0.90cm]{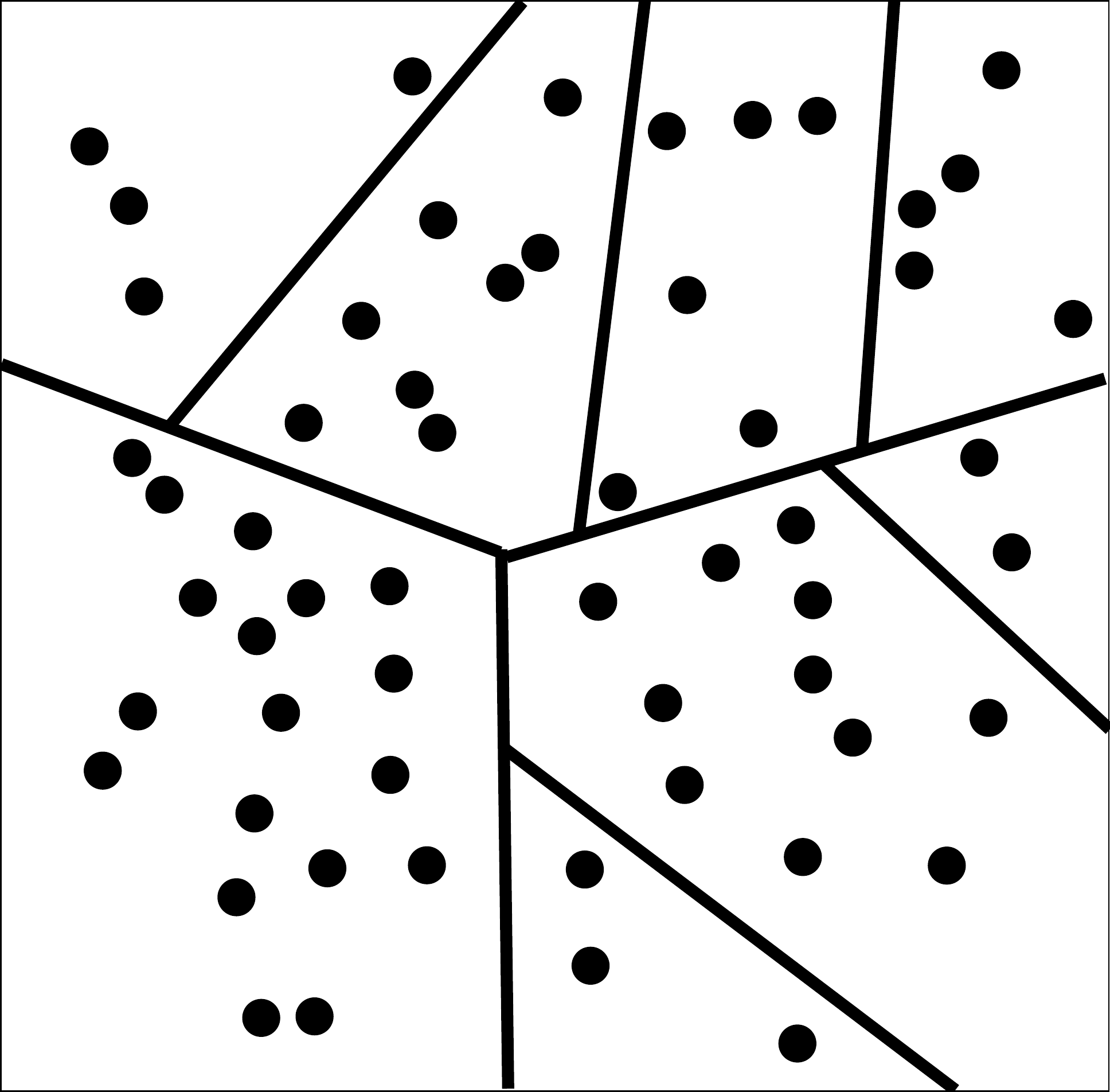}};
\draw[draw, rounded corners] (0,-7) rectangle (15,-6);
\draw[draw] (7, -7) to (7, -6);
\draw[draw] (8, -7) to (8, -6);
\node[] at (7.5, -5.6) {$T_{K, i}$};

\draw[draw] (1, -1.5) to (1, -0.5);
\draw[draw] (1, -3) to (1, -2);
\draw[draw] (2, -3) to (2, -2);
\draw[draw] (3, -3) to (3, -2);

\node[] at (11.5, 1.5) {Pointset $S$};

\node[] at (11.5, 0) {\includegraphics[width=2cm]{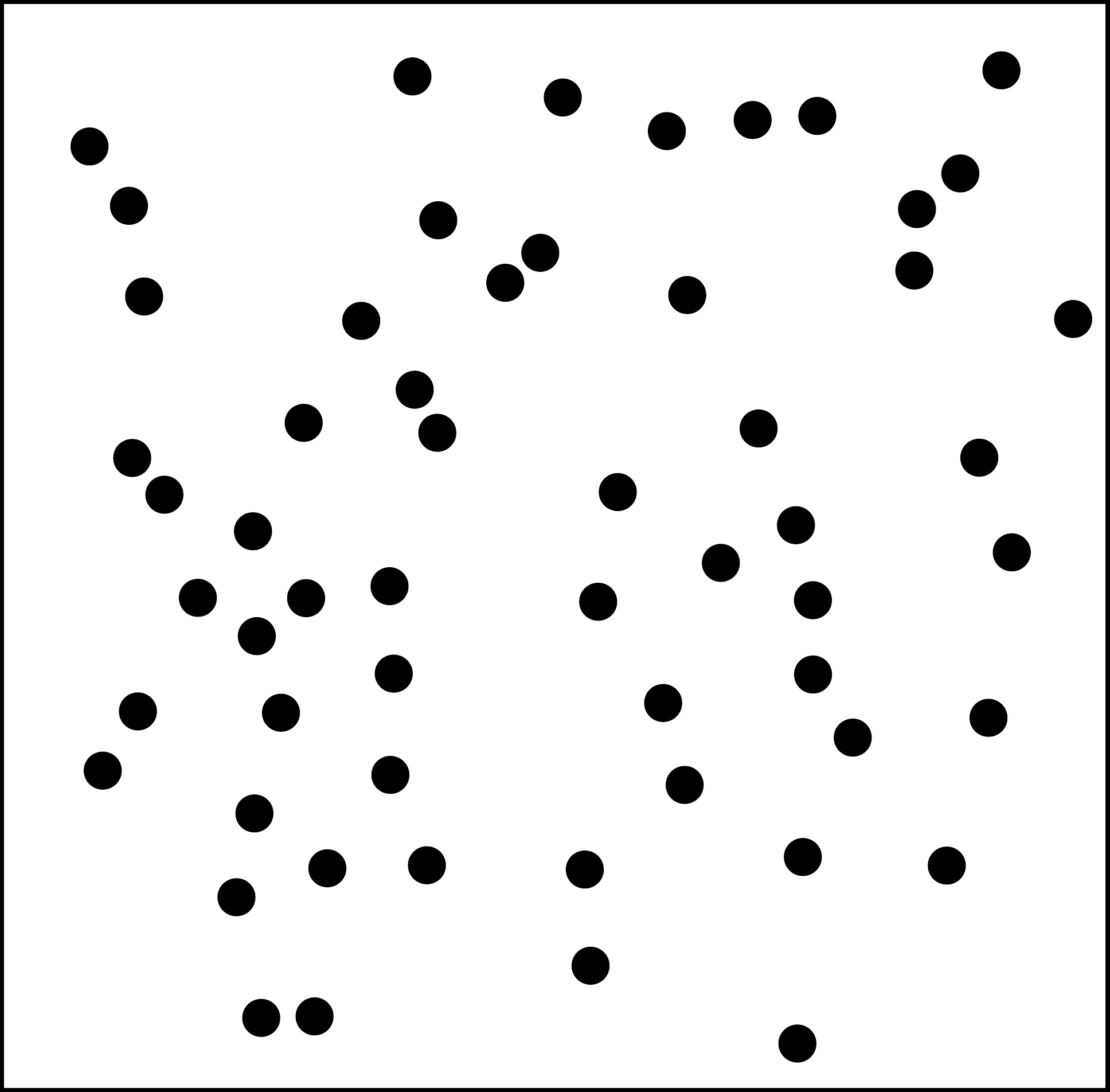}};

\draw [->] (10.3, 0) to[out=170, in=20] (1.1, 1); 
\draw [->] (10.3, -.1) to[out=170, in=20] (0.5, -0.45); 
\draw [->] (10.3, -.2) to[out=170, in=20] (1.5, -0.45); 
\draw [->] (10.3, -1) to[out=190, in=20] (7.8, -5.9); 

\node[] at (6, 1.5) {$h_{0,1}$};
\node[] at (5, 0.7) {$h_{1,1}$};
\node[] at (4.5, -0.1) {$h_{1,2}$};
\node[] at (9.8, -3) {$h_{K,i}$};

\foreach \x in {1,...,5}{
    \draw[draw] (\x, -7) to (\x, -6);
}

\node[] at (10.2, -.5) {\tiny{$\vdots$}};

\draw[draw] (13.8,-7) to (13.8, -6);

\node[] at (0.5,0.5) {$T_{0, 1}$};
\node[] at (0.5,-1) {$T_{1, 1}$};
\node[] at (1.5,-1) {$T_{1, 2}$};
\node[] at (0.5,-2.5) {$T_{2, 1}$};
\node[] at (1.5,-2.5) {$T_{2, 2}$};
\node[] at (2.5,-2.5) {$T_{2, 3}$};
\node[] at (3.5,-2.5) {$T_{2, 4}$};

\node[] at (-1, 1.5) {Level};
\node[] at (-1, 0.5) {$0$};
\node[] at (-1, -1) {$1$};
\node[] at (-1, -2.5) {$2$};
\node[] at (-1, -4) {$\vdots$};
\node[] at (-1, -6.5) {$K$};

\node[] at (16, 1.5) {Rep.};
\node[] at (16, 0.5) {${p^{-0}}$};
\node[] at (16, -1) {${p^{-1}}$};
\node[] at (16, -2.5) {${p^{-2}}$};
\node[] at (16, -6.5) {${p^{-K}}$};

\node[] at (2, -4) {$\vdots$};
\node[] at (5.5, -4) {$\ddots$};

\foreach \x in {1,...,5}{
    \node at (\x - 0.5, -6.5) {$T_{K, \x}$};
}
\footnotesize{\node at (14.4, -6.5) {$T_{K, p^{-K}}$};}

\end{tikzpicture}
\caption{Overview of the multi-level LSH data structure with tables $T_{k, i}$
and hash functions $h_{k, i}$ splitting a data set $S$. The data structure 
is set up for levels $0, ..., K$
with repetition count ``Rep.'' for collision probability $p$ (avoiding ceilings).
Example for a space partition of $S$ induced by hash function $h_{K, i}$
is explicitly depicted as the content of table $T_{K, i}$ where each class is a
bucket.}
\label{fig:datastructure}
\end{figure}
We describe an alternative tree-based data structure that trades 
query time for space consumption in Appendix~\ref{sec:trie}.
We remark that some of the proposed query algorithms will require a slightly higher repetition
count.
In such cases, the function $\text{reps}$ will be redefined.
The additional repetition count will never add more than a polylog overhead to the running time.

\section{Standard LSH, Local Expansion, and Probing the Right Level}\label{sec:results}
In this section we show that using a standard LSH approach might 
yield running time
$\Omega(t n^\rho)$ when standard parameter settings such as the ones from
Theorem~\ref{thm:lsh} are used to solve SRR. Then, we define a measure
for the difficulty of a query. Finally, we 
show that if the output size and this measure is known, 
inspecting a certain level in the multi-level LSH data structure gives 
output- and query-sensitive expected running times.

\begin{observation}
    Suppose we want to solve SRR in $(X, \dist)$ using LSH with parameters as in Theorem~\ref{thm:lsh} with LSH family $\HH$.
    Let $q \in X$ be a fixed query point. Then there exist data sets $S
    \subseteq X$ with $|S| = n$ such that the expected number of points retrieved from the hash
    tables on query $q$ in the LSH data structure built for $S$ is $\Omega(t n^\rho)$.
\end{observation}
The basic idea of the proof is to inspect  
difficult data sets in which there are $t - 1$ very close points
to the query, one point at distance $r$, and all the other points at distance
close to $cr$. Details are deferred to
Appendix~\ref{sec:worstcase}.

For a set $S$ of points, a point $q$, and a number $r > 0$, let $N_r(q)$ be the number of points in $S$ at distance at most $r$ from $q$.
We next define the \emph{expansion} at a query point $q$ for a given distance.
The expansion measures how far we can increase the radius of an $r$-sphere around the query point before the number of points covered increases above some constant factor.
This dimensionality measure is central in the running time analysis of our proposed algorithms.
\begin{definition}[Expansion]
    Let $r > 0$, $q \in X$ and $S \subseteq X$ be a set of points.
    The \emph{expansion} $c^\ast_{q,r}$ at point $q$ is the largest number $c$ such that $N_{cr}(q) \leq 2 N_r(q)$, where
    $c^\ast_{q,r}$ is $\infty$ if $N_r(q) \geq n/2$.
\end{definition}
We will often simply write $c^*_q$, when $r$ is known in the context. A visualization for the expansion around a query is
depicted in Figure~\ref{fig:example}.

\subsection{Query Algorithms If $t$ and $c^\ast_q$ are Known}
\label{sec:static:analysis}
\begin{theorem}
    Let $r > 0$ and $c \geq 1$. Let $S$ be a set of $n$ points and let DS be the Multi-level LSH data structure obtained from preprocessing $S$ with $L = \Omega(n^{\rho(r,c)})$.
    Given a query point $q$, let $t = N_r(q)$ and $c_q^\ast$ be the expansion around $q$ in $S$.
    There exists a query algorithm on DS to solve SRR with the following properties:
    \begin{enumerate}
        \item[(i)] If $c_q^\ast \geq c$, the algorithm has expected 
            running time $O(t (n/t)^{\rho(r, c_q^\ast)})$. 
        \item[(ii)] Otherwise, the algorithm has expected running time
            $O(t(n/t)^{\rho(r,c)} + N_{cr}(q))$. 
    \end{enumerate}
    For $t=0$, the running time is the same as $t=1$.
    \label{thm:overview:simple}
\end{theorem}
\begin{proof}
    Let $p_1$ and $p_2$ be the probabilities that the query point collides with points at distance $r$ and $c_q^\ast r$, respectively, given the LSH family used.
    We consider statement (i) first.
    Look at level $k = \lceil \frac{\log (n/t)}{\log (1/p_2)} \rceil$ in DS such that $p_2^k\le t/n$ and $p_1^{-k}=\Theta((n/t)^{\rho(r, c^\ast_q)})$.
    Since $c^*_q\ge c$, we can assume $p_1^{-k}\le L$, so that we may inspect that many repetitions and guarantee constant collision probability with a close points.
    The total expected number of collisions is then not more than
  $p_1^{-k}(t + N_{c_q^*r}(q) p_1^k + n p_2^k)$.
 By the choice of $k$, $np_2^k \le t$ and so this is $O(tp_1^{-k} + N_{c_q^*}(q))$.
 By the definition of $c^*_q$, $N_{c_q^*r}(q) = O(t)$ and so this term is dominated by the former.
 Finally looking at every bucket takes time $O(p_1^{-k})$, but this is likewise dominated if $t\ge 1$.
 Statement (ii) follows by the same line of reasoning, simply using $c$ instead of $c^*_q$. Since this value 
 of $c$ does not have the expansion property, the term $N_{cr}(q)$ is present in the running time.
\end{proof}
As can be seen from the theorem statement, our running time bounds might depend on the number of points at
distance at most $cr$. This happens when the expansion around the query is
smaller than the $c$ value that can be read off from the number of repetitions
and the LSH hash family at hand. The influence of these points is however only
linear in their number. 
The result basically shows that there exists a single level of the multi-level LSH data structure that we want to probe when 
$t$ and the expansion around the query is known.

\section{Adaptive Query Algorithms}\label{sec:algorithm}
In this section we describe a query algorithm that obtains the results from Theorem~\ref{thm:overview:simple} without knowing $t$ or the expansion around the query.
It turns out that we can get something even better, i.e., a running time equivalent to knowing the entire distribution of distances $\text{dist}(q,x)$ from the query point $q$ to data points $x$ in the data set.

We work on a multi-level LSH data structure, DS, set up for $S\subseteq X$ with tables $T_{k, i}$.
DS is assumed to have been built with $L$ repetitions and $K$ levels, see Definition~\ref{def:multi:level:lsh}.
Now, a query algorithm looking at the buckets at level $k$, would be expected to take time
\begin{align}
   \E[W_k]
      &= p_1^{-k} (O(1) + \sum_{x \in S}\Pr[h_k(q)=h_k(x)]).
\end{align}
This accounts for the $p_1^{-k}$ repetitions to be made for correctness (cf.~Appendix~\ref{sec:proof:lsh}), $O(1)$ time for probing and evaluating the hash functions, and, by linearity of expectation, the number $\sum_{x \in S}\Pr[h_k(q)=h_k(x)]$ of expected collisions and thus points retrieved from the hash tables.
For ease of presentation, we omit ceilings for repetition counts $p_1^{-k}$ and note that adding a constant to $k$ never changes the results by more than a constant factor.

The function $\E[W_k]$ over $k$ will have an optimum in $[0,K]$.
It turns out that we can indeed find this optimum and get expected running time close to:
\begin{align}
   W_\text{single}
   &=\min_{0 \le k \le K} \E[W_k]
  \label{eq:w:single}.
\end{align}

Since this describes the best expected running time given knowledge of the distance distribution \emph{at the query}, we note that the quantity is always upper bounded by running times stated in
Theorem~\ref{thm:overview:simple}. (Given that the number of repetitions is sufficiently high as stated there.)
However, in many important cases it can be much smaller than that!
In Appendix~\ref{sec:examples:w:single} we calculate $W_\text{single}$ for different data distributions, including ``locally growth-restricted'' as considered in~\cite{datar2004locality}.
In this case it turns out that $W_\text{single} = O(\log n)$, an exponential improvement over ``standard'' query time $O(n^\rho)$.

The query algorithm is given as Algorithm~\ref{algo:adaptive:query} and works as follows:
For each level $0 \leq k \leq K$, calculate the work of doing $\text{reps}(k) = p_1^{-k} (2\log 2k)$ repetitions by summing up all bucket counts.
(The $2\log 2k$ factor is a technical detail explained in the proof below.)
Terminate as soon as the optimal level has been provably found, which may be one that we have considered in the past, and report
all close points in the candidate set. This decision is based on whether the number of buckets to look at alone is larger than 
the smallest candidate set found so far or not.

\begin{algorithm}[t]
    \caption{Adaptive-Single-Probe($q$, $p_1$, $T$)}\samepage\label{algo:adaptive:query}
    \begin{algorithmic}[1]
   \State $k \gets 1, k_{\text{best}} \gets 0, w_{k_{\text{best}}} \gets n;$ 
   \While{$\text{reps}(k) \le \min(L, w_{k_{\text{best}}})$}
   \State $w_k \gets \sum_{i = 1}^{\text{reps}(k)} (1 + |T_{k,i}(q)|)$;
      \If{$w_k < w_{k_{\text{best}}}$}
         \State $k_{\text{best}} \gets k; w_{k_{\text{best}}} \gets w_k$;
      \EndIf
      \State $k \gets k + 1$;
   \EndWhile
   \State\Return $C \gets \bigcup_{i = 1}^{\text{reps}(k_{\text{best}})} \{x \in T_{k_{\text{best}}, i}(q) \mid \text{dist}(x,q)\le r\}$
    \end{algorithmic}
\end{algorithm}
\begin{theorem}
    Let $S \subseteq X$ with $|S| = n$ and $r$ be given.
    Then Algorithm~\ref{algo:adaptive:query} on DS solves SRR with constant probability.
    The expected running time of the \textbf{while}-loop in Lines~(2)--(6)
    and the expected number of distance computations in Line~(7) is
    $O(W_\text{single} \log\log W_\text{single})$.
    \label{thm:adaptive:time}
\end{theorem}

\emph{Proof.}
    First we show that the algorithm works correctly, then we argue about its running time.

    For correctness, let $y \in S$ be a point with $\text{dist}(y, q) \leq r$.
    At each level $k$, we see $y$ in a fixed bucket with probability at least $p_1^k$. With $\text{reps}(k)$ repetitions,
    the probability of finding $y$ is at least $1-(1-p_1)^{\text{reps}(k)} \ge 1 - 1/(2k)^2$.
    By a union bound over the $K$ levels of the data structure, $y$ is present on every level with probability at least $1-\sum_{k=1}^\infty 1/(2k)^2 \ge 1/2$, which shows correctness. \matodo{I removed a line here.}

    Now we consider the running time. The work inside the loop is dominated by Line~(3) which takes time $O(\text{reps}(k))$, given constant access to the size of the buckets.
    Say the last value $k$ before the loop terminates is $k^\ast$, then the loop takes time
    $\sum_{k = 1}^{k^*} O(\log k \cdot p_1^{-k})
    \le \log k^* \cdot p_1^{-k^*} \sum_{k = 0}^{\infty} O(p_1^k)
    = O(\log k^* \cdot p_1^{-k^*}) = O(w_{k_{\text{best}}})$, where the last equality is by the loop condition, $\text{reps}(k^\ast)\le w_{k_\text{best}}$.

    In Line~7, the algorithm looks at $w_{k_\text{best}}$ points and buckets.
    Hence the total expected work is
    \begin{align*}
       \E(w_{k_{\text{best}}})
       &= \E\left[\min_{0 \leq k \leq K} w_k\right]
       &&\le \min_{0 \leq k \leq K} \E[w_k] &&\text{by Jensen's inequality}
       \\&=\min_{0\le k\le K} \log k \cdot \E[W_k]
       &&\le \log k' \cdot \E[W_{k'}] &&\text{where $k'=\argmin_{0\le k\le K}\E[W_k]$}
       \\&= O(W_\text{single}\log\log W_\text{single}) &&&&\text{as $p_1^{-k'} \le W_\text{single}$.\tag*{\qed}}
    \end{align*}

\subsection{A Multi-probing Version of Algorithm~\ref{algo:adaptive:query}}\label{sec:algorithm:multiprobing}

We extend the algorithm from the previous subsection to take advantage of multi-probing.
For a particular hash function $h_k$, distance $r$ and value $\lambda\ge1$, we define
    a \emph{probing sequence} $\sigma = (\sigma_{k,\ell})_{1 \leq \ell \leq \lambda}$ as a sequence of functions $R^k\to R^k$.
Now when we would probe bucket $T_{k, i}(h_k(q))$, we instead probe $T_{k, i}(\sigma_{k,1}h_k(q)), \dots, T_{k, i}(\sigma_{k,\lambda}h_k(q))$.
(Where $\sigma_{k,1}$ will usually be the identity function.)

For a point $y$ at distance $r$ from $q$, we will be interested in the event $[\sigma_{k,\ell}h_k(q)=h_k(y)]$. The probability that this event occurs is donated by $p_{k,\ell}$. 
If $p_{k,1} \ge p_{k,2} \ge \dots$, we say that the probing sequence is \emph{reasonable}.
The intuition is that we probe buckets in order of expected collisions. In particular, by disjointness of the events, the probability of a collision within the first $\ell$ probes is exactly $P_{k, \ell} = p_{k, 1} + \dots + p_{k, \ell}$.
Hence doing $\ell$ probes per repetition, we need about $1/P_{k,\ell}$ repetitions to obtain
constant probability of finding $y$.
In practice, probing sequences are usually not reasonable \cite{Andoni15}, but as long as $p_{k,\ell}$ is known, they can be sorted in advance. 

To state the complexity of our algorithm, we generalize the quantity $W_\text{single}$ from \eqref{eq:w:single} in the natural way to include multi-probing.
As in the case of probing a single bucket, $W_\text{multi}$ denotes the minimal work one would expect to need for an LSH based approach that knows the optimal values of $k$ and $\ell$.
\begin{align}
   W_\text{multi}
   &=\min_{0 \leq k \leq K, 1 \leq \ell} \left[
   \frac1{P_{k, \ell}}
\left(\ell + \sum_{x \in S,1 \leq i\le \ell} \Pr[\sigma_{k, i}h(q) = h(x)]\right)\right]
      \label{eq:w:multi}
\end{align}

As in Algorithm~\ref{algo:adaptive:query}, we carefully explore the now two-dimensional and infinite space $[0,K]\times[1,\infty]$ of parameters.
Consider the probability of finding some point $y$ at distance $r$ to our query.
Say we choose values $(k,\ell)$ and make $\text{reps}(k, \ell)$ repetitions; we would then probe buckets
$\bigcup_{i=1}^{\text{reps}(k,\ell)}\bigcup_{j=1}^{\ell} T_{k, i}[\sigma_{k,j}h_k(q)]$ and find $y$ with probability $1 - (1-P_{k,\ell})^{\text{reps}(k,\ell)}$.
As in the single-probing algorithm, we have to be careful about dependencies and 
set 
   $\text{reps}(k, \ell) = (2 \log(2\ell k))/P_{k, \ell}$
such that the probability of not finding $y$ is less than $\exp(-\text{reps}(k,\ell)/P_{k,\ell}) = (2\ell k)^{-2}$.
A union bound over the whole parameter space yields $\sum_{k=1}^\infty\sum_{\ell=1}^\infty(2\ell k)^{-2} < 7/10$, and so $y$ is present for every parameter choice with constant probability.

All that remains is reusing the idea from Algorithm~\ref{algo:adaptive:query} of maintaining upper and lower bounds on the final work, and stop once they meet. The parameter space is explored using a priority queue. 
The pseudocode of the algorithm is given as Algorithm~\ref{algo:adaptive:query:multiprobe}.
To obtain good query time it is necessary to store the values $W_{k, \ell}$ computed so far and reuse them in Line~\ref{algo:multiprobe:comp} of Algorithm~\ref{algo:adaptive:query:multiprobe}.
Details are given in the proof below.
\begin{algorithm}[t]
   \caption{Adaptive-Multi-probe($q$, $\sigma$, $T$)}\samepage
   \label{algo:adaptive:query:multiprobe}
   \begin{algorithmic}[1]
      \State $W_{\text{best}} \gets n$; $k_\text{best} \gets 0$;
      $\ell_\text{best} \gets 1$; $\text{PQ} \gets $ empty priority queue
      \For{$1 \leq k \leq K$}
         \State$\text{PQ.insert}((k, 1), \text{reps}(k, 1))$
      \EndFor
      \While{PQ.min() $< W_{\text{best}}$}\label{algo:multiprobe:begin:while}
         \State $(k, \ell) \gets$ PQ.extractMin()
         \State PQ.{insert}($(k, \ell + 1)$, $(\ell+1)\cdot$reps$(k, \ell + 1)$)
         \State $W_{k,\ell} \gets \sum_{i = 1}^{\text{reps}(k,\ell)}\sum_{j = 1}^{\ell}
         \left(1+|T_{k, i}[\sigma_{k,j}h_{k,i}(q)]|\right)$\label{algo:multiprobe:comp}
         \If {$W_{k,\ell} < W_{\text{best}}$}
            \State $k_\text{best} \gets k;\quad \ell_\text{best} \gets \ell;\quad W_{\text{best}} \gets W_{k,\ell}$
         \EndIf
         \EndWhile\label{algo:multiprobe:end:while}
      \State\Return $
            \bigcup_{i=1}^{\text{reps}(k_\text{best},\ell_\text{best})}
            \bigcup_{j=1}^{\ell_\text{best}}
            \{x \in T_{k_\text{best}, i}[\sigma_{k_\text{best},j}h_{k_\text{best},i}(q)] \mid \text{dist}(x,q)\le r\}$\label{algo:multiprobe:return}
   \end{algorithmic}
\end{algorithm}

\begin{theorem}
   Let $S \subseteq X$ and $r$ be given. Let $(k,\ell)$ be a pair that  
   minimizes the right-hand side of \eqref{eq:w:multi}.
   Given a reasonable probing sequence $\sigma$, Algorithm~\ref{algo:adaptive:query:multiprobe} on DS solves SRR. If DS supports at least $\text{reps}(k, \ell)$ repetitions,  
   the expected running time is $O(W_\text{multi} \log^3 W_\text{multi})$ and 
   the expected number of distance computations is $O(W_\text{multi}\log W_\text{multi})$.
   \label{thm:adaptive:multiprobe:time}
\end{theorem}

\begin{proof}
   Correctness of the algorithm follows by the explanation before the theorem statement. 

   To analyze the running time, we define $\text{cost}(k,\ell) = \ell\cdot \text{reps}(k,\ell)$. This value is used as the priority of a parameter pair
   in the priority queue.
   It provides a lower bound on the work required to consider this parameter pair because we have to check that many buckets.
    Next we note that it cannot happen that all pairs $(k, \ell)$ in the priority queue have cost larger than $W_{\text{best}}$, but there exists a pair $(k', \ell')$ with $k' \geq k$ and $\ell' \geq \ell$ not inspected so far such that $\text{cost}(k', \ell') < W_\text{best}$.
    This is because for fixed $k$ the cost is non-decreasing in $\ell$ by Lemma~\ref{lem:monsum} (provided in Appendix~\ref{sec:monsum}) and for fixed $\ell$ cost$(k, \ell)$ is non-decreasing in $k$.

    To compute a new value $W_{k, \ell + 1}$ in Line~\eqref{algo:multiprobe:comp} of the algorithm, we take
    advantage of the work $W_{k, \ell}$ already discovered, and only compute the
    number of buckets that are new or no longer needed. Specifically, we compute 
    \begin{align*}
        W_{k, \ell+1} = W_{k, \ell}
        + \sum_{i=1}^{\text{reps}(k,\ell+1)} |T_{k, i}(\sigma_{k, \ell+1}(q))|
        - \sum_{j=1}^{\ell} \sum_{i=1 + \text{reps}(k,\ell)}^{\text{reps}(k,\ell+1)}
      |T_{k, i}(\sigma_{k, j}(q))|.
    \end{align*}
    So, for each $k$ we never visit a bucket more than twice and amortized over all operations, the computation of $W_{k, \ell + 1}$ takes time $O(\text{reps}(k, \ell + 1))$.
    For each $k$ with $1 \leq k \leq K$, let $\ell_k^\ast$ be the largest value
    of $\ell$ such that the pair $(k, \ell)$ was considered by the algorithm.
    The total cost of computing $W_{k, 1}, \ldots, W_{k, \ell_k^\ast}$ is then at most $\sum_{i = 1}^{\ell^\ast} \text{reps}(k,i)
    = \sum_{i=1}^{\ell^*}2(\log(2ki))/P_{k,i}
    = O(\ell^*(\log \ell^\ast)^2/P_{k, \ell^\ast})$
    by Lemma~\ref{lem:monsum}.
    By the loop condition, we know that $\log \ell^\ast \cdot \ell^\ast/P_{k, \ell^\ast}$ is at most $W_\text{best}$, so the algorithm spends time $O(W_\text{best} \log W_\text{best})$ for fixed $k$.

    Let $k_\text{max}$ be the maximum value of $k$ such that a pair $(k, \ell)$ was considered by the algorithm.
    Since we stop when every item in the priority queue has a priority higher or equal to $W_\text{best}$, we must have $k_\text{max}\le\frac{\log W_\text{best}}{\log1/p_1}$ because we need at least $p_1^{-k_{\text{max}}}$ repetitions for the single probe on level $k_\text{max}$.

    Thus, the final search time for the while-loop is $O(W_\text{best}
    \log(W_\text{best})^2)$ and the algorithm makes exactly $W_\text{best}$ distance computations in Line~\eqref{algo:multiprobe:return}.
    Now observe that by bounding the additional repetitions by $O(\log W_\text{multi})$, the same reasoning as 
    in the proof of Theorem~\ref{thm:adaptive:time} shows $\E(W_\text{best}) = O(W_\text{multi} \log W_\text{multi})$. 
\end{proof}
%
%

\section{A Probing Sequence in Hamming Space}\label{sec:hamming:multiprobing}
In this section we analyze bitsampling LSH in Hamming space~\cite[Section 3.2.1]{Har-PeledIM12} using a novel, simple probing sequence.
We consider the static setting as in Section~\ref{sec:static:analysis}, where the number of points to report and the expansion around the query is known.
We then show the existence of certain (optimal) level and probing length parameters, and prove that using those give a good expected running time.
The adaptive query algorithm from Section~\ref{sec:algorithm} would find parameters at least as good as those, and thus yield a running time as least as good as what we show here.

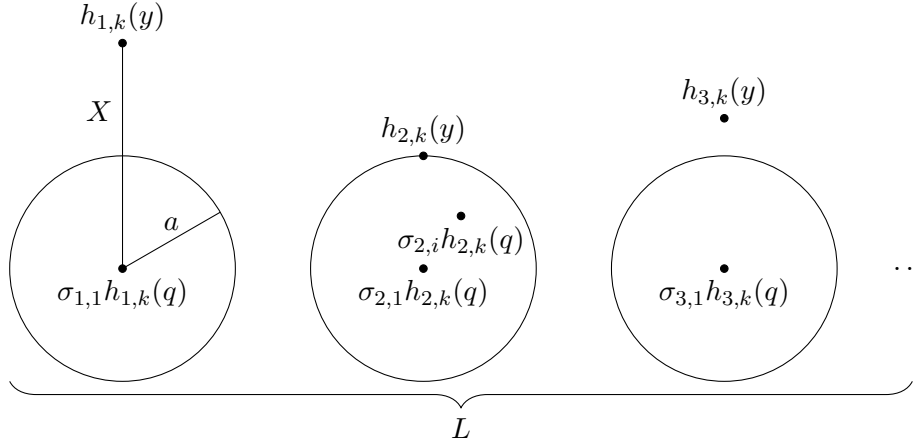
\begin{figure}[t!]
   \centering
   \begin{tikzpicture}
      \draw (-4,0) -- ++(30:1.5) node[pos=.5,above] {$a$};
      \draw (-4,0) -- (-4,3) node[pos=.7,left] {$X$};
      \draw (-4,0) circle (1.5);
      \draw (0,0) circle (1.5);
      \draw (4,0) circle (1.5);
      \draw[fill=black] (-4,0) circle (.05) node[below] {$\sigma_{1,1}h_{1,k}(q)$};
      \draw[fill=black] (0,0) circle (.05) node[below] {$\sigma_{2,1}h_{2,k}(q)$};
      \draw[fill=black] (.5,.7) circle (.05) node[below] {$\sigma_{2,i}h_{2,k}(q)$};
      \draw[fill=black] (4,0) circle (.05) node[below] {$\sigma_{3,1}h_{3,k}(q)$};
      \draw[fill=black] (-4,3) circle (.05) node[above] {$h_{1,k}(y)$};
      \draw[fill=black] (0,1.5) circle (.05) node[above] {$h_{2,k}(y)$};
      \draw[fill=black] (4,2) circle (.05) node[above] {$h_{3,k}(y)$};
      \node at (6.5,0) {$\dots$};
      \draw [decorate,decoration={brace,amplitude=10}] (6.5,-1.5) -- (-5.5,-1.5)
         node[midway,anchor=north,yshift=-10] {$L$};
   \end{tikzpicture}
   \caption{At each of the $L$ repetitions we query the closest $\ell$ positions.
   Since the projected distance $X$ to our target point $y$ is distributed as $\text{Bin}(k,\text{dist}(q,y)/d)$, we find $y$ with constant probability by setting $L=O(\Pr[X\le a]^{-1})$.}
   \label{fig:ham_repeat}
\end{figure}

Our scheme uses hash functions $h_k : \{0,1\}^d \rightarrow \{0,1\}^k$ that sample $k$ positions at random with repetition.
For a fixed query point $q \in \{0,1\}^d$ and $k \geq 1$,
the probing sequence $\sigma_{k, \ell}$ maps $h_k(q)$ to the $\ell$-th closest point in $\{0,1\}^k$, where ties are broken arbitrarily.
This sequence can be generated efficiently, see \cite{knuth2011combinatorial}.

Fix a target close point $y \in \{0,1\}^d$ at distance $r$,
let $p$ be the probability that $q$ and~$y$ collide,
and let $p_{k, \ell}$ be the probability that $y$ lands in the $\ell$-th bucket that we check.
Furthermore, let $V(a)=\sum_{i=0}^{a}{k\choose i}$ be the volume of the radius $a$ hamming ball.
If $\sigma_{k,\ell}h_k(q)$ is at distance $a$ to $h(q)$, we have a collision if $q$ and $y$ differ in exactly $a$ out of the $k$ coordinates chosen by $h_k$.
Hence, $p_{k,\ell} = p^{k-a}(1-p)^a$ for the $a$ satisfying $V(a-1) < \ell \le V(a)$.
Thus, the sequence is reasonable. 
Figure~\ref{fig:ham_repeat} illustrates our approach.

The best approximations to sizes of hamming balls are based on the entropy function.
Hence, for the purpose of stating the theorem, we introduce the following notation. 
For $\alpha \in [0,1]$ and $\beta \in [0,1]$, we let
$H(\alpha) = \alpha \log 1/\alpha + (1 - \alpha) \log 1/(1 - \alpha)$ and
$\infdiv{\alpha}{\beta} = \alpha \log (\alpha / \beta) + (1 - \alpha) \log ((1
- \alpha)/(1 - \beta))$ denote the binary entropy of $\alpha$ and the relative
entropy between $\alpha$ and $\beta$, respectively.
Moreover, 
let $\rho(r,c) = \tfrac{\log t}{\log n}\left(1+\tfrac{\infdiv{\alpha}{1-p(r)}}{\HE(\alpha)}\right)$
where $\alpha$ is defined implicitly from $\tfrac{\log t}{\log n}\left(1+\tfrac{\infdiv{\alpha}{1-p(cr)}}{\HE(\alpha)}\right)=1$.
\begin{theorem}
   Let $r > 0$ and $c \geq 1$.
   Let $S$ be a set of $n$ points and let DS be the Multi-level LSH data structure obtained from preprocessing $S$ with $L = \Omega(n^{\rho(r,c)})$.
   Given a query point $q$, let $t = N_r(q)+1$ and let $c_q^\ast$ be the expansion around $q$ in $S$. If $c_q^\ast \geq c$,
   there exists a query algorithm on DS to solve SRR with running time $O(n^{\rho(r, c_q^*)})$, otherwise
   the running time is $O(n^{\rho(r,c)}+N_{cr}(q))$.
   \label{thm:multi_static}
\end{theorem}
We do not know of a simple closed form for $\rho(r,c)$, but Figure~\ref{fig:work/logn} on Page~\pageref{fig:work/logn} shows a numerical evaluation for comparison with the running time obtained for single-probing and other approaches.

The figure suggests that we always get better exponents than the single-probe approach, and that we get optimal query time for large $t$ and asymptotically optimal for $t=n^{o(1)}$.
Corollary~\ref{cor:multi_static} confirms this:
\begin{corollary}
   Let $\rho=(\log p_1)/(\log p_2) < 1/c$ be the usual exponent for bitsampling, then:

   If $t \ge n^{\frac{\HE(1-p_2)}{\HE(1-p_2)+\infdiv{1-p_2}{1-p_1}}}$,
   the expected query time is $O(n^{\rho(r,c)}) = O(n^\rho + t)$.

   If $t=n^{o(1)}$, the expected query time is
   $O(n^{\rho(r,c)}) = O(n^\rho t^{o(1)})$.

   \label{cor:multi_static}
\end{corollary}

\begin{proof}[Proof of Theorem~\ref{thm:multi_static}]
   We will now show that the smaller number of repetitions needed by multi-probing leads to fewer collisions with the $t$ very close points in hard instances of SRR.
   To see this, we bound the value of $W_\text{multi}$ from \eqref{eq:w:multi} as follows:
   \begin{align*}
      W_\text{multi}
      &=
      \min_{k, \ell} \left[
            \frac1{\sum_{1\le \ell \le \lambda}\Pr[\sigma_{k, \ell}(h_k(q)) = h_k(y)]}
            \left(
               \ell
               + \!\!\!\!\!\sum_{x \in S,1 \leq \ell \le \lambda} \!\!\!\Pr[\sigma_{k, \ell}(h_k(q)) = h_k(x)]
            \right)
      \right]
      \\&\le
      \min_{k, a} \left[
            \frac1{\Pr[\text{dist}(h_k(q),h_k(y))\le a]}
            \left(
               V_k(a)
               + \sum_{x \in S} \Pr[\text{dist}(h_k(q),h_k(x))\le a]
            \right)
      \right]
      \\&\le
      \min_{k, a} \left[
            \frac1{\Pr[\text{Bin}(k,1\text{-}p_1)\le a]}
               \left(
                  V_k(a)
                  + t
                  + t'\Pr[\text{Bin}(k,1\text{-}p_1)\le a]
                  + n\Pr[\text{Bin}(k,1\text{-}p_2)\le a]
               \right)
      \right]
   \end{align*}
   The first inequality holds by restricting $\ell$ to only take values that are the volume of a $k$-dimensional hamming ball; in the second inequality we upper bounded the collision probabilities for points in ranges $[0,r)$, $[r,cr)$ and $[cr,d]$.

   The next step is to minimize this bound over the choice of $k$ and $a$.
   We focus on $t'=O(t)$ and so we want $V_k(a) = t = n\Pr[\text{Bin}(k,1\text{-}p_2)\le a]$.
   For simplicity we write $\alpha=a/k$ for the normalized radius.
   We use the following tight bound on the tail of the binomial distribution:
   $\Pr[\text{Bin}(k,p)\le \alpha k] = \exp(-k\infdiv{\alpha}{p})\Theta(1/\sqrt{k})$ for $\alpha\in(0,1/2)$
   and $V_k(\alpha k) = \exp(k H(\alpha))\Theta(1/\sqrt k)$~\cite{petrov2012sums}.
   Then our equation can be written as $k \HE(\alpha) = \log t = \log n - k\infdiv{\alpha}{1-p_2}$.
   This suggests $k=\frac{\log t}{\HE(\alpha)}=\frac{\log n}{\HE(\alpha) + \infdiv{\alpha}{1-p_2}}$ and $\frac{\HE(\alpha) + \infdiv{x}{1-p_2}}{\HE(\alpha)} = \frac{\infdiv{\alpha}{1-p_2}}{\HE(\alpha)} + 1$.
   We can then plug $k$ into the bound on $W_\text{multi}$:
   \begin{align}
      W_\text{multi}
      &\le \frac{3t}{\Pr[\text{Bin}(k,1-p_1)\le a]\Theta(\sqrt k)} + t'
      \notag\\&= O(t \exp(k \infdiv{\alpha}{1-p_1})) + t'
      = O\left(n^{\frac{\log t}{\log n}\left(\frac{\infdiv{\alpha}{1-p_1}}{\HE(\alpha)} +1\right)}\right) + t'
      \label{eqn:log_static_work}
   \end{align}
   which are exactly the values stated in Theorem~\ref{thm:multi_static}.
\end{proof}

\begin{proof}[Proof Sketch Corollary~\ref{cor:multi_static}]
   For the first statement observe that if $\alpha$ is as large as $1-p_1$, then $\Pr[\text{Bin}(k,1-p_1)\le\alpha k]$ is constant.
   The second factor in the minimization has all terms being within a constant of $t$, and so the whole thing becomes $O(t)$.
   We can check that $\alpha\ge1-p_1$ happens exactly when $\frac{\log t}{\log n} \ge\frac{\HE(1-p_2)}{\HE(1-p_2)+\infdiv{1-p_2}{1-p_1}}$.
   In this range $t \ge n^\rho$, so $O(t)=O(n^\rho+t)$.

   For the second part of the corollary, we solve the equation implied by Theorem~\ref{thm:multi_static}, asymptotically as $\tau=\frac{\log t}{\log n}\to0$.
   Details can be found in Appendix~\ref{app:cor:proof}, but the idea is as follows:
   We first define $f_p(\alpha) = 1+\frac{\infdiv{\alpha}{p}}{\HE(\alpha)}$, and show $f_{p_1}(\alpha)=(\rho+\psi\alpha/\log\frac1{p_2}+O(\alpha^2))f_{p_2}(\alpha)$ for~$\psi$ being 
   the constant defined in Corollary~\ref{cor:multi_static}.
   Using bootstrapping, we show the inversion $\alpha = f_{p_2}^{-1}(1/\tau) = \frac{\log1/p_2}{\alpha\log1/\alpha}+O(1/\log\frac1\alpha)$.
   Plugging this into \eqref{eqn:log_static_work} proves the corollary.
\end{proof}

\section{Conclusion}\label{sec:conclusion}
In this article we proposed an adaptive LSH-based algorithm for Spherical Range Reporting that is never worse than a static LSH data structure knowing optimal parameters for the query in advance, and much better on many input distributions where the output is large or the query is easy.

The main open problem remaining is to achieve target time $O(n^\rho+t)$ for all inputs.
One approach might be a data-dependent data structure as described in \cite{andoni2015optimal}. In the light of our multi-probing results, we however wonder if the bound can be obtained data independently as well. 
Here, it would be interesting to analyze other probing sequences. It would be interesting to see whether one can describe
adaptive query algorithms that make use of the output-sensitive space/time-tradeoff data structures we described in 
Appendix~\ref{sec:appendix2}.
Finally, it would be natural to extend our methods to give better LSH data structures for the approximate $k$-nearest neighbor problem.


\bibliography{lit}


\appendix

\section{Proof of Theorem~\ref{thm:lsh}}
\label{sec:proof:lsh}

\begin{proof}
Given access to an LSH family $\HH$ with the properties stated in the theorem
and two parameters $L$ and $k$ (to be specified below), repeat the following process independently for
each $i$ in $\{1, \ldots, L\}$:
Choose $k$ hash functions $g_{i, 1},\ldots,g_{i, k}$ independently at random from~$\HH$.
For each point $p \in S$, we view the
sequence $h_i(p) = (g_{i, 1}(p),\ldots, g_{i, k}(p)) \in R^k$ as the hash code of $p$,
identify this hash code with a bucket in a table, and store a reference to $p$
in bucket $h_i(p)$.
To avoid storing empty buckets from $R^k$, we resort to
hashing and build a hash table $T_i$ to store the non-empty buckets for $S$ and $h_i$.

Given a query $q \in X$, we retrieve all points from the buckets $h_1(q),
\ldots, h_L(q)$ in tables $T_1, \ldots, T_L$, respectively, and report
a close point in distance at most $cr$ as soon as we find such a point.
Note
that the algorithm stops and reports that no close points exists after
retrieving more than $3L$ points, which is crucial to guarantee query time
$O(n^\rho)$.

The parameters $k$ and $L$ are set according to the following reasoning.
First,
set $k$ such that it is expected that at most one distant point at distance at
least $cr$ collides with the query in one of the repetitions.
This means that
we require $n
p_2^k \leq 1$ and hence we define $k = \lceil \frac{\log n}{\log(1/p_2)} \rceil$.
To find a close point at distance at most $r$ with probability at
least $1 - \delta$, the number of repetitions $L$ must satisfy
$\delta \leq (1 - p_1^k)^L \leq \exp(-p_1^k \cdot L)$.
This means that $L$
should be at least $p_1^{-k} \ln \delta$ and simplifying yields $L = O(n^\rho)$.
Note that these parameters are set to work even in a worst-case scenario where there is
exactly one point at distance $p$ and all other points have distance slightly larger 
than $cr$.
\end{proof}

\section{A Trie-based Version of the Data Structure}
\label{sec:trie}
In this section we discuss an alternative representation of our data structure. This
is meant as a replacement for the Multi-level LSH data structure described in Definition~\ref{def:multi:level:lsh}. It offers better space consumption while being slower to query. 

As in the LSH forest data structure proposed by Bawa et al.~\cite{Bawa05}, we do not 
store references to data points in hash tables. Instead we use a sorted array with a trie as a 
navigation structure on the array. The technical description follows.
   
First, choose $K \cdot L$ functions $g_{i, j}$ for $1 \leq i \leq L$ and $1 \leq
k \leq K$ from $\HH$ independently at random.  For each $i \in \{1, \ldots,
L\}$, we store a sorted array $A_i$ with references to all data points in $S$
ordered lexicographically by there bucket code over $R^K$. To navigate this
array quickly, we build a trie over the bucket codes of all keys in $S$ of depth
at most $K$. Each vertex of the trie has two attributes {\tt leftIndex} and {\tt
rightIndex}. If the path from the root of the trie to vertex $v$ is labeled
$L(v)$, then {\tt leftIndex} and {\tt rightIndex} point to the left-most and
right-most elements in $A_i$ whose bucket code starts with $L(v)$. We fix some
more notation. For each point $q \in X$, we let $v_{i, k'}(q)$ be the vertex in
trie $\TT_i$ that is reached by searching for the bucket code of $q$ on level at
most $k'$. Furthermore, we let $\TT_{i, k}(q)$ denote the set of keys that 
share the same length $k$ prefix with $q$ in trie $\TT_i$. We can compute 
$|\TT_{i, k}(q)|$ by subtracting $v_{i,k}(q).\texttt{leftIndex}$ from
$v_{i,k}(q).\texttt{rightIndex} + 1$.

\section{Difficult Inputs for Standard LSH}
\label{sec:worstcase}
Suppose we want to solve SRR using an LSH family $\HH$.
Assume that the query point $q \in X$ is fixed.
Given $n$, $t$ with $1 \leq t \leq n$, and $c > 1$, we generate a data set $S$ by picking
\begin{itemize}
\item $t - 1$ points at distance $\epsilon$ from $q$, for $\epsilon$ small enough that even concatenating $\lceil\frac{\log n}{\log 1/p_2}\rceil$ hash functions from $\HH$, 
    we still have collision probability higher than $0.01$,
\item one point $x \in X$ with $\dist(q, x) = r$,
\item the remaining $n-t$ points at distance $cr$.
\end{itemize}
We call a set $S$ that is generated by the process described above a \emph{$t$-heavy input for SRR on~$q$}. By 
definition, a $t$-heavy input has expansion $c$ at query point $q$.
We argue that the standard LSH approach is unnecessarily slow on such inputs.
\begin{observation}
    Suppose we want to solve SRR in $(X, \dist)$ using LSH with parameters as in Theorem~\ref{thm:lsh} with LSH family $\HH$.
    Let $q \in X$ be a fixed query point, and $S$ be a $t$-heavy input generated by the process above.
    Then the expected number of points retrieved from the hash tables on query $q$ in the LSH data structure is $\Theta(t n^\rho)$.
\end{observation}

\begin{proof}
    The standard LSH data structure is set up with $k = \lceil \frac{\log n}{\log 1/p_2}\rceil$ and $L=O(n^\rho)$. 
    $L$ repetitions are necessary to find the close point at distance $r$ with constant probability.
    By the construction of $S$, each repetition will contribute at least $\Theta(t)$ very close points in expectation.
    So, we expect to retrieve $O(tn^\rho)$ close points from the hash tables in total.
\end{proof}
    The process described above assumes that the space allows us to pick sufficiently many points at a certain distance.
    This is for example true in $\mathbb{R}^d$ with Euclidean distance.
    In Hamming space $\{0, 1\}^d$ we would change the above process to enumerate the points from distance $1, 2, \ldots$ and distance $cr + 1, cr + 2, \ldots$.
    If $d$ and $r$ are sufficiently large, the same observation as above also holds for inputs generated according to this process.

\section{Examples For Calculating $W_\text{single}$ for Certain Input Distributions}
\label{sec:examples:w:single}

In this section we discuss two examples to get a sense for quantity \eqref{eq:w:single} defined
on Page~\pageref{eq:w:single}.

\paragraph{Example 1 (Random Points in Hamming Space)}
Fix a query point $q \in \{0,1\}^d$ and assume that our data set $S$ consists of $n$ uniform random points from $\{0,1\}^d$.
Then the distance $X$ from our query point is binomially distributed $\sim\text{Bin}(n,1/2)$.
If we choose bitsampling as in~\cite{IndykM98} as hash function, $\sum_{x\in S}\Pr[h_k(q)=h_k(x)]$ is just $n\E((1-X/d)^k)=nd^{-k}\E(X^k)$.
This coresponds to finding the $k$th moment of a binomial random variable, which we can approximate by writing $X = d/2 + Z_d\sqrt{d/4}$ where $\E(Z_d)=0$ and $Z_d\to Z$ converges to a standard normal.
Then
$n\E((X/d)^k)
= n2^{-k}\E(1 + Z_d/\sqrt{d})^k
= n2^{-k}(1 + k\E(Z_d)/\sqrt{d} + O(k^2/d))
= n2^{-k}(1+O(k^2/d))$.
For dimension $d=\Omega(\log n)^2$ our algorithm would take $k\approx\log_2 n$ to get
$\sum_{x\in S}\Pr[h_k(q)=h_k(x)] = O(1)$ and
$W=n^{\frac{\log 1/p_1} {\log 2}}$.
Just as we would expect for LSH with bitsampling and far points at distance $cr=d/2$.

\paragraph{Example 2 (Locally Growth-Restricted Data)} Another interesting setting to consider is when the data is locally
growth-restricted, as considered by Datar et al. \cite[Appendix
A]{datar2004locality}.
This means that the number of points within distance $r$ of
$q$, for any $r>0$, is at most $r^c$ for some small constant $c$.
In \cite{datar2004locality}, the LSH framework is changed by providing the parameter $k$ to the
hash function. However, if we fix $r = k$, our algorithm will find a candidate
set of size
$W=O(\log n)$.
So, our algorithm takes advantage of restricted
growth and adapts automatically on such inputs.

The proof from \cite{datar2004locality} works,
since they also inspect all colliding points.  It is easy to see that the integral
$\int_1^{r/\sqrt{2}}e^{-Bc}c^b \text{ d}c$ is still bounded by $2^{O(b)}$ when we start
at $c=0$ instead of $c=1$, since the integrand is less than 1 in this interval.

\section{Lemma~\ref{lem:monsum}}
\label{sec:monsum}

\begin{lemma}
   \label{lem:monsum}
   Let $x_1 \ge x_2 \ge \dots$ be a non-increasing series of real numbers.
   Let $X_n = \sum_{k=1}^n x_k$ be the $k$th prefix sum.
   Then it holds:
   \begin{align}
      n/X_n
         &\le (n+1)/X_{n+1} \label{eq:rect} \\
      \sum_{k = 1}^n 1/X_k
         &= O(n\log n/X_n).
      \label{eq:total}
   \end{align}
\end{lemma}
Here we have used the approximation for harmonic numbers, $H_n = 1 + 1/2 + \dots + 1/n = \log n + O(1)$, by Euler~\cite{euler1740progressionibus}.
\begin{proof}
   Since the values $x_k$ are non-increasing, we have
   $X_n \ge n x_n \ge n x_{k+1}$
   and so
   \begin{align*}
      (n + 1) X_n
      \ge n X_n + n x_{n+1}
      = n X_{n+1}
   \end{align*}
   which is what we want for \eqref{eq:rect}.
   For the second inequality, we use \eqref{eq:rect} inductively, we get $a/X_a \le b/X_b$ whenever $a\le b$.
    Hence we can bound \eqref{eq:total} term-wise as
    \begin{align*}
       \sum_{k = 1}^n \frac1{X_k}
       \stackrel{\eqref{eq:rect}}{\le} \sum_{k = 1}^n \frac{n}{k X_n}
        = \frac{n}{X_n} \sum_{k = 1}^n \frac{1}{k}
        = \frac{n}{X_n} H_n
        = O(n\log n/X_n).
    \end{align*}
\end{proof}
We may notice that the bound is tight for $x_1=x_2=\dots=x_n$.
Say $x_k=1$ for all $k$, then $X_k=k$ and $\sum_{k=1}^n1/X_k = H_n = \Omega(n\log n/X_n)$.
In the other extreme, $\sum_{k=1}^n1/X_k\ge n/X_n$, which is sharp when $x_1=1$ and $x_k=0$ for $k\ge2$.

\section{Proof of Corollary~\ref{cor:multi_static}, second part}\label{app:cor:proof}
When $t$ is small compared to $n$, the multiprobing radius $\alpha$ can be made smaller.
In this regime, we hence consider the following expansion:
\begin{align*}
   f(\alpha,p_1)
   &=
   1+\frac{\infdiv{\alpha}{1-p_1}}{\HE(\alpha)}
   \\&=
   \frac{\HE(\alpha)+\infdiv{\alpha}{1-p_1}}{\HE(\alpha)+\infdiv{\alpha}{1-p_2}}
   f(\alpha,p_2)
   \\&=
   \frac{\log\frac1{p_1} + \alpha\log\frac{p_1}{1-p_1}}
   {\log\frac1{p_2} + \alpha\log\frac{p_2}{1-p_2}}
   f(\alpha,p_2)
   \\&=
   \left(\frac{\log p_1}{\log p_2}
      + \frac{\log\frac{p_1}{1-p_1}\log\frac1{p_2} - \log\frac1{p_1}\log\frac{p_2}{1-p_2}}
      {\left(\log\frac1{p_2}\right)^2}\alpha + O(\alpha^2)
   \right)f(\alpha,p_2)
   \\&=
   (\rho + \psi \alpha/\log{1/p_2} + O(\alpha^2))f(\alpha,p_2),
   \numberthis \label{eqn:ff_asymp}
\end{align*}
for constants $\rho$ and $\psi$ depending on $p_1$ and $p_2$.
This already gives us that we are asymptotically optimal, as long as $\alpha f(\alpha,p_2)$ goes to 0 as $f(\alpha,p_2)$ goes to $\infty$.
To see that this is indeed the case, we need the following asymptotics:
\begin{align*}
   \HE(\alpha) + \infdiv{\alpha}{1-p}
   &= \alpha \log\tfrac1{1-p} + (1-\alpha)\log\tfrac1p
   \\&= \log\tfrac1p + O(\alpha)
   \\
   \HE(\alpha)
   &= \alpha \log\tfrac1\alpha + (1-\alpha)\log\tfrac1{1-\alpha}
   \\&= \alpha \log\tfrac1\alpha + (1-\alpha)(\alpha-O(\alpha^2))
   \\&= \alpha (\log\tfrac1\alpha + 1) + O(\alpha^2)
   \\
   f(\alpha,p)
   &= (\HE(\alpha) + \infdiv{\alpha}{1-p})/\HE(\alpha)
   \\&= \frac{\log\frac1p}{\alpha(\log\frac1\alpha+1)} + O(1/\log\tfrac1\alpha)
   \numberthis \label{eqn:fx_asymp}
\end{align*}
We would like invert~\eqref{eqn:fx_asymp} to tell us how fast $\alpha$ goes to
zero, and plug that into~\eqref{eqn:ff_asymp}.
To this end, we let $y=f(\alpha,p_2)/\log\tfrac1{p_2}$.
Then it is clear that, at least asymptotically, $1/y^2<\alpha<1/y$.
That tells us $\alpha=y^{-\Theta(1)}$, and we can use this estimate to ``bootstrap'' the inversion:
\begin{align*}
   \alpha
   &=
   \frac1{y(\log\frac1\alpha+1)} + O\left(\frac \alpha{y\log\tfrac1\alpha}\right)
   \\&=
   \frac1{y\left(\log\left[\frac1{\left(
               \frac1{y(\log\frac1\alpha+1)} + O(1/(y^2\log\tfrac1\alpha))
   \right)}\right]+1\right)} + O(1/(y^2\log\tfrac1\alpha))
   \\&=
   \frac1{y\left( \log\left[y(\log\frac1\alpha+1)\right] + \log\left[\frac1{1+O(1/y)}\right] + 1\right)} + O(1/y^2)
   \\&=
   \frac1{y\log y + O(y\log\log y)} + O(1/y^2)
   \\&=
   \frac1{y\log y} + O\left(\frac{\log\log y}{y(\log y)^2}\right)
   \numberthis \label{eqn:inv_asymp}
\end{align*}
Plugging the result back into~\eqref{eqn:ff_asymp} we finally get:
\begin{align*}
    \log\E(W_k)
   &= (\log t) f(\alpha,p_1)
   \\&= \log t \left(\rho + \tfrac{\psi}{(\log{1/p_2})y\log y} + O\left(\frac{\log\log y}{y(\log y)^2}\right)\right) f(\alpha,p_2)
   \\&=
   \log t \left(\rho f(\alpha,p_2)
      + \frac{\psi}{\log f(\alpha,p_2)}
   + O\left(\frac{\log\log f}{(\log f)^2}\right)\right)
   \\&=
   \rho \log n
   + \log t \left(\frac{\psi}{\log\frac{\log n}{\log t}}
   + O\left(\frac{\log\log\frac{\log n}{\log t}}
      {\left(\log\frac{\log n}{\log t}\right)^2}\right)\right),
\end{align*}
as $\frac{\log n}{\log t}$ goes to $\infty$, i.e., 
$\tau$ goes to $0$.

\section{A Different Approach to Solving SRR with LSH}
\label{sec:appendix2} We reconsider the approach to solve SRR presented in Indyk's Ph.D.
thesis~\cite[Page~12]{Indyk00} under the name ``enumerative PLEB''. 
While his method does not yield good running times directly, it is possible to
combine a number of very recent results, to get running times similar to the ones
achieved by our methods. We give a short overview of this approach next. As in 
Section~\ref{sec:static:analysis}, we assume that the number of points $t$ to 
report is known. At the end of this section we describe a counting argument that 
is also contained in Indyk's Ph.D. thesis \cite{Indyk00} that allows to solve 
the $c$-approximate spherical range counting problem in an output-sensitive way. 

Indyk describes a black-box reduction to solve SRR using a standard dynamic data
structure for the $(c,r)$-near neighbor problem.
It works by repeatedly querying an $(c,r)$-near point data structure (time
    $O(n^{\rho_q})$) and then deleting the point found (time
$O(n^{\rho_u})$), where $\rho_q$ and $\rho_u$ are the query- and
update-parameters. (For a standard LSH approach, we have $\rho_q = \rho_u$.)
This is done until the data structure no longer reports any points within
distance $r$.
Due to the guarantees of an $(c,r)$-near neighbor data structure, in the worst
case the algorithm recovers all points within distance $cr$, giving a total
running time of $t'(n^{\rho_q}+n^{\rho_u})$, where $t'$ is the number of points
within distance $cr$ which might yield a running time of $\Omega(n^{1+\rho})$ as 
noticed in Appendix~\ref{sec:worstcase}. 

Of course, we can never guarantee sublinear query time when $t'$ is large, but
we can use a space/time-tradeoff-aware to improve the factor of $t$, when the number of returned points is large.

We will assume the $(c,r)$-near neighbor data structure used in the reduction is based on LSH.
In~\cite{andoni2016lower}, Andoni et al. describe a general data structure comprising
loosely ``all hashing-based frameworks we are aware of'':
\begin{definition}[List-of-points data structure]
    \mbox{}
\begin{itemize}
   \item Fix sets $A_i\subseteq \mathcal{R}^d$, for $i=1\ldots m$; with each
   possible query point $q \in \mathcal{R}^d$, we associate a set of indices
   $I(q) \subseteq [m]$ such that $i \in I(q) \Leftrightarrow q \in A_i$;
   \item For a given dataset $S$, the data structure maintains $m$ lists of
     points $L_1, L_2, \dots, L_m$, where $L_i=S\cap A_i$.
\end{itemize}
\end{definition}
Having such a data structure, we perform queries as follows: For a query point $q$, 
we scan through each list $L_i$ for $i \in I(q)$ and check whether there exists some $p \in L_i$ with $\|p - q\| \leq cr$. If it exists, return $p$.

Data structures on this form naturally allow insertions of new points, and we
notice that if ``Lists'' are replaced by ``Sets'' we can also efficiently perform updates. 

To solve spherical range reporting, we propose the following query algorithm for
a point $q$:
\begin{enumerate}
   \item For each $i \in I(q)$ look at every point $x$ in $L_i$.
   \item If $\|x-q\| \le r$, remove the point from all lists, $L_j$, where it is present. 
\end{enumerate}
This approach allows for a very natural space/time-tradeoff. Assuming that
querying the data structure takes expected time $O(n^{\rho_q})$ and updates take
expected time $O(n^{\rho_u})$, the expected running time of the query is 
$O(n^{\rho_q} + t n^{\rho_u})$. This asymmetry can be exploited with a time/space tradeoff.
In very recent papers~\cite{laarhoven2015tradeoffs,Christiani16,andoni2016lower}
it was shown how to obtain such tradeoffs in Euclidean space for approximation 
factor $c \geq 1$, for any pair
($\rho_q, \rho_u$)
that satisifies
\begin{align*}
c^2 \sqrt{\rho_q} + (c^2 - 1) \sqrt{\rho_u} = \sqrt{2c^2 - 1}.
\end{align*}
To minimize running time, we may take exponents balancing $T = n^{\rho_q} = t n^{\rho_u}$ and obtain
\begin{align*}
   \frac{\log T}{\log n} &= \frac{1}{2 c^2-1}
   + \frac{c^2-1}{2 c^2-1}\tau
   + \frac{c^2 \left(c^2-1\right) }{2 c^2-1}\left(2-\tau-2 \sqrt{1-\tau}\right)
   \stackrel{\text{(*)}}\le \rho + (1-c^4\rho^2)\tau,
\end{align*}
where $\tau=\frac{\log t}{\log n}$ and $\rho = 1/(2c^2-1)$.
Here (*) holds for $t \le \frac{2 c^2-1}{c^4}$, and $T=O(t)$ otherwise. Note that
this approach requires knowledge of $t$. A visualization of the running time guarantees of this
approach is shown in Figure~\ref{fig:indyk}. Note that it requires knowledge
of $t$ and does not adapt to the expansion around the query point. It would 
be interesting to see whether our adaptive methods could be used to obtain 
a variant that is query-sensitive. Next, we discuss an algorithm for the 
spherical range counting problem that can be used to obtain an 
approximation of the value 
$t$ sufficient for building the data structure presented here. 

\begin{figure}[t]
\centering
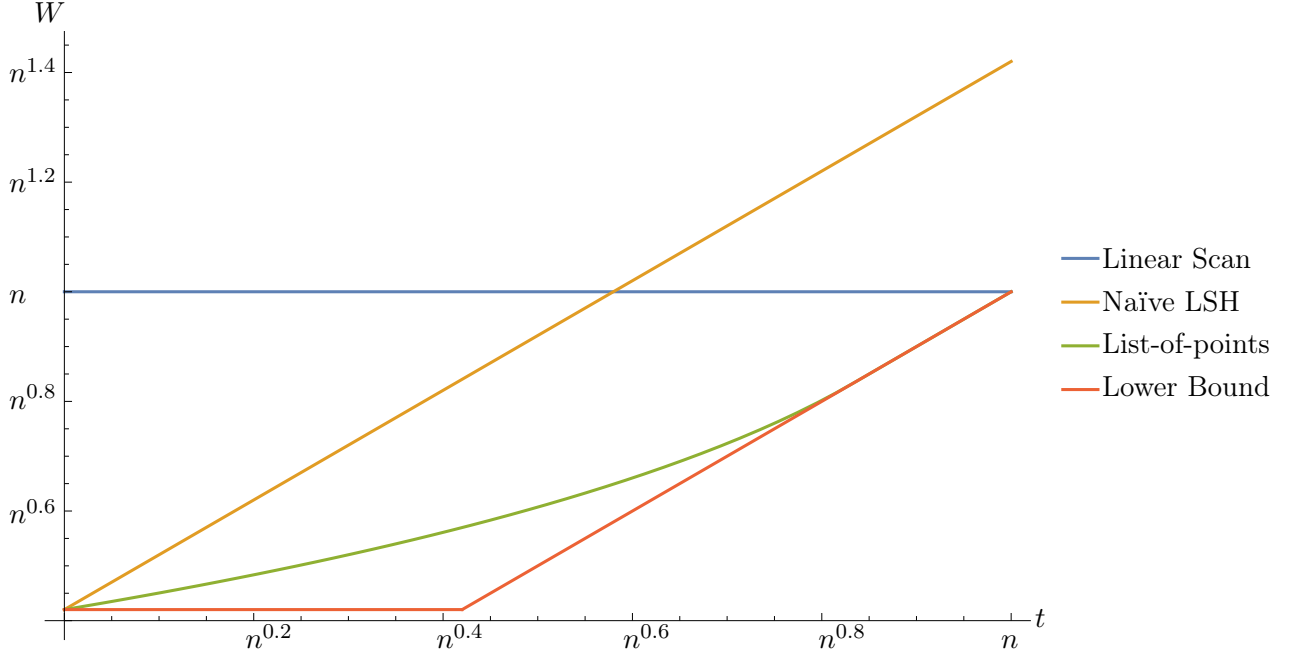
\caption{Visualization of the running time guarantees of
    the space/time-tradeoff list-of-points data structure 
    for $c=1.3$ in Euclidean space.
    The $x$-axis shows the value of $t$ compared to $n$, the $y$-axis
    shows the expected work $W$.
    For comparison, we plotted the lower bound of $O(n^\rho + t)$, the running
    time $O(t n^\rho)$ of the naïve LSH approach, and the running time $O(n)$ of
a linear scan.}
\label{fig:indyk}
\end{figure}

\subsection{Solving $c$-approximate Spherical Range Counting}
In \cite[Chapter 3.6]{Indyk00}, Indyk shows that by performing $O((\log n)^2/\alpha^3)$ queries to 
independently built $(c,r)$-near neighbor data structures, there is an algorithm
that returns for a query $q$ a number $C$ such that $(1-\alpha)N_r(q) \leq C
\leq (1 + \alpha) N_{cr}(q)$ with constant probability. The running time of the 
black-box reduction is $O(n^\rho (\log n)^2/\alpha^3)$. We show in this
section that we can solve the problem in time $O((n/t)^\rho \log n/\alpha^3)$.

At the heart of the algorithm of \cite{Indyk00} is a subroutine that has 
the following output behavior for fixed $C$:
\begin{enumerate}
   \item If $N_{cr}(q) \le C(1-\alpha)$, it will answer SMALLER
   \item If $N_{r}(q) \ge C$, it will answer GREATER
\end{enumerate}
The subroutine uses $O(\log n/\alpha^2)$ queries of independently build 
$(c,r)$-near neighbor data structures, each built by sampling $n/C$ points from the data set. 

We can use the above subroutine to solve the spherical range counting problem in
time $O((n/t)^\rho \log n/\alpha^3)$ time as follows.
Half the size of $\alpha$, and perform a geometrical search
for the values $t=n, (1-\alpha)n, (1-\alpha)^2n,\ldots$\,. Assuming that a query on a data structure that
contains $n$ points takes expected time $O(n^\rho)$ and stopping as soon as the
algorithm answers ``Greater'' for the first time, we obtain a running time
(without considering the $O(\log n/\alpha^2)$ repetitions for each $t$ value) of
\begin{align*}
    &\left(\frac{n}{n}\right)^\rho
   + \left(\frac{n}{n(1-\alpha)}\right)^\rho
   + \left(\frac{n}{n(1-\alpha)^2}\right)^\rho + \dots
   + \left(\frac{n}{t}\right)^\rho
   \\&\le
   \left(\frac{n}{t}\right)^\rho
   + \left(\frac{n(1-\alpha)}{t}\right)^\rho
   + \left(\frac{n(1-\alpha)^2}{t}\right)^\rho
   + \dots
   \\&=
   \left(\frac{n}{t}\right)^\rho \frac1{1-(1-\alpha)^\rho}
   \le
   \left(\frac{n}{t}\right)^\rho \frac1{\alpha\rho},
\end{align*}
which results in a total running time of $O((n/t)^\rho \log n/\alpha^3)$. 


\end{document}